\documentclass[11pt,nofootinbib , notitlepage, onecolumn, a4paper,tightenlines,prd,longbibliography]{revtex4-1}
\usepackage[UKenglish]{babel}
\usepackage[utf8]{inputenc}
\usepackage{lmodern}
\usepackage[T1]{fontenc}
\usepackage{textcomp}
\usepackage{amsmath,amssymb,lmodern,amsthm}
\usepackage{mathtools}
\usepackage{graphicx}
\usepackage{mathrsfs}
\usepackage{hyperref}
\usepackage{url}
\usepackage{array}
\usepackage{tensor}
\usepackage{pifont}
\usepackage[toc,page]{appendix}
\usepackage{graphicx}
\usepackage[left=2.5cm,right=2.5cm,top=3cm,bottom=3cm]{geometry}
\usepackage{enumitem}
\usepackage{empheq}
\usepackage{cleveref}

\theoremstyle{plain}
\usepackage{xargs}
\usepackage{xcolor}
\hypersetup{
	colorlinks=true,
	urlcolor=blue,
	linkcolor={violet!50!black},
	citecolor={green!40!black}
}
\newtheorem{thm}{Theorem}[]
\theoremstyle{definition}

\newtheorem{deff}{Definition}[]

\newtheorem{remark}{Remark}
\renewcommand{\iff}{\Longleftrightarrow}
\newcommand{\liff}{\Longleftrightarrow}
\newcommand{\lied}{\pounds}
\newcommand{\defeq}{\ \vcentcolon=\ }

\newcommand{\defeqc}{\ \stackrel{\Sc}{\vcentcolon=}\ }
\newcommand{\scri}{\mathscr{J}}

\newcommand{\pt}[2]{\tensor{\hat{#1}}{#2}}
\newcommand{\ctr}[3]{\tensor[#1]{#2}{#3}}
\newcommand{\ct}[2]{\tensor{#1}{#2}}

\newcommand{\ctp}[2]{\tensor{\overline{#1}}{#2}}
\newcommandx{\ps}[3][2=, 3=]{\hat{#1}^{#2}_{#3}}
\newcommandx{\cs}[3][2=, 3=]{#1^{#2}_{#3}}

\newcommand{\eqs}{\ \stackrel{\scri}{=}\ }
\newcommand{\eqc}{\ \stackrel{\Sc}{=}\ }

\newcommand{\cd}[1]{\tensor{\nabla}{#1}}
\newcommand{\cds}[1]{\tensor{\overline{\nabla}}{#1}}
\newcommand{\cdc}[1]{\tensor{D}{#1}}

\newcommand{\commute}[2]{\left[#1,#2\right]}

\newcommand{\Q}{\mathcal{Q}}
\newcommand{\T}{\mathcal{T}}
\newcommand{\D}{\mathcal{D}}
\newcommand{\W}{\mathcal{W}}
\newcommand{\Z}{\mathcal{Z}}
\newcommand{\Sc}{\mathcal{S}}
\newcommand{\Pc}{\mathcal{P}}

\newcolumntype{M}[1]{>{\centering\arraybackslash}m{#1}}
\newcolumntype{N}{@{}m{0pt}@{}}

\def\be{\begin{equation}}
\def\ee{\end{equation}}
\def\bea{\begin{eqnarray}}
\def\eea{\end{eqnarray}}
\def\bean{\begin{eqnarray*}}
\def\eean{\end{eqnarray*}}

\makeatletter
\def\Dated@name{\normalsize} 
\let\@fnsymbol\@roman
\makeatother

\begin{document}

\title{\Large A novel characterisation of gravitational radiation in asymptotically flat space-times}
	
\author{Francisco Fernández-Álvarez}
\email[Email address: ]{francisco.fernandez@ehu.eus} 
\author{José M. M. Senovilla} 
\email[Email address: ]{josemm.senovilla@ehu.eus}
\affiliation{Departamento de Física Teórica e Historia de la Ciencia,\\ Universidad del País Vasco UPV/EHU\\ Apartado 644, 48080 Bilbao, Spain }

\date{\today{}}

	\begin{abstract}
			\begin{center}
			\rule{\linewidth/4}{1pt}
			\end{center}
	 A novel criterion to determine the presence of gravitational radiation arriving to, or departing from, null infinity of any weakly asymptotically-simple space-time with vanishing cosmological constant is given. The quantities involved are geometric, of tidal nature, with good gauge behaviour and univocally defined at null infinity. The relationship with the classical characterisation using the News tensor is analysed in detail. A new balance law at infinity is presented, which may be useful to define `radiation states'.
	\end{abstract}
\maketitle
	
	\section{Introduction and setting}
		The confidence on the existence of gravitational waves received a tremendous boost during the 1950-60's with the works of Trautman \cite{Trautman58}, Pirani  \cite{Pirani57}, Bel \cite{Bel1962}, Bondi \cite{Bondi1962}, Sachs \cite{Sachs1962}, Newman \cite{Penrose62} and others (see e.g. \cite{Zakharov}). A robust, covariant approach to the gravitational radiation and the structure of the conformal boundary (timelike, spacelike, and null infinity) of asymptotically flat space-times was then developed by Penrose \cite{Penrose65,Penrose68} and nicely formulated by Geroch \cite{Geroch1977}. Furthermore, a description of the radiative degrees of freedom in terms of the intrinsic geometry of null infinity was given by Ashtekar \cite{Ashtekar81}. As a main outcome of those works the presence of gravitational radiation `escaping' from --or `entering' into-- the space-time was successfully characterised by means of the so-called {\em News tensor}.
		
		However, this formalism only applies to the case with vanishing cosmological constant and we wonder if an alternative description of the radiation at infinity is feasible; in particular, one that may perform equally well in the presence of a positive cosmological constant too. The aim of this communication is to present our proposal, based on the Bel-Robinson tensor \cite{Bel1958}, for such an alternative in the well-established asymptotically flat case; the corresponding situation for a positive cosmological constant will be considered in a subsequent paper \cite{Fernandez-Alvarez_Senovilla}. 

 From a physical point of view, the News tensor carries the information about the energy-momentum radiated away by an isolated system, while the Bel-Robinson tensor describes the energy-momentum of the {\em tidal} gravitational field ---for historical reasons, one uses the name `super-energy' for this. It is a fundamental fact, due to the Equivalence Principle, that there does not exist any notion of {\em local, pointwise} gravitational energy-momentum. There are, however, several notions of quasilocal energy-momentum, usually associated to closed 2-dimensional surfaces \cite{Szabados2004}. On the other hand, the super-momentum constructed from the Bel-Robinson tensor is local and its vanishing is unambiguous at any point of the space-time, stating whether or not the tidal forces vanish there ---in vacuum. Given that actual measurements of gravitational waves are basically of tidal nature, it seems like a good idea to explore the Bel-Robinson tensor as a viable object detecting the existence of gravitational radiation. This is our main argument, and we are going to prove that it does work. Actually, the relationship between super-energy and quasilocal energy has been analyzed for many years \cite{Horowitz82,Senovilla2000a,Szabados2004} in vacuum.		 
		
		Throughout the paper, we work in the conformal completion $ (\cs{M},\ct{g}{_{\alpha\beta}}) $ with boundary $\scri $ --null infinity-- of any weakly asymptotically simple space-time $ (\ps{M},\pt{g}{_{\alpha\beta}}) $ with zero cosmological constant $ \Lambda=0 $ (for a detailed description of these completions see, e.g., \cite{Frauendiener2004} or \cite{Kroon}).\footnote{Our signature is $ (-,+,+,+) $ and the Riemann tensor is defined by $ \commute{\cd{_\alpha}}{\cd{_\beta}}\ct{v}{_\gamma} = \ct{R}{_{\alpha\beta\gamma}^\mu} \ct{v}{_\mu}$, where $ \cd{_\alpha} $ is the covariant derivative on $ \left(\cs{M},\ct{g}{_{\alpha\beta}} \right) $. Also, $R_{\alpha\gamma}=\ct{R}{_{\alpha\mu\gamma}^\mu}$ and $R=g^{\alpha\gamma} R_{\alpha\gamma}$.}  Both metrics are related by $ \ct{g}{_{\alpha\beta}} = \Omega^2\pt{g}{_{\alpha\beta}} $ on $\ps{M}$, where the conformal factor $ \Omega $ is a function on $M$ satisfying  $\Omega>0$ on $\ps{M}$, $\Omega=0$ at $\scri$ and $n_\alpha :=\cd{_\alpha}\Omega\neq0$ at $\scri$. There exists a conformal rescaling freedom, $ \Omega \rightarrow \omega\Omega $ with $ \omega>0$ everywhere.  We fix partially this freedom by considering only those conformal factors satisfying $ \cd{_\alpha}\ct{n}{_\beta}\eqs 0 $. 
		The remaining allowed transformations are those preserving the condition $ \lied_{\vec{n}} \omega \eqs 0$, which will be the only ones considered from now on. 
		
		One can prove that the one-form $ \ct{n}{_\alpha}$, clearly normal to $\scri$, is null there $n_\alpha n^\alpha  \eqs 0$, and in our conventions $ \ct{n}{^\alpha} $ is future pointing. Hence, $\scri$ is a lightlike hypersurface whose first fundamental form, denoted by $\bar g_{ab}:= g_{\mu\nu}e^\mu{}_a e^\nu{}_b$, is degenerate (we use Latin indices on $ \scri $, $ a, b,\dots =1,2,3 $). Here $ \{\ct{e}{^\alpha_a}\} $, defined {\em only} on $\scri$,  is a basis of the set of vector fields tangent to $\scri$, $n_\alpha \ct{e}{^\alpha_a}=0$ ; obviously $\ct{n}{^\alpha} \eqs \ct{n}{^a}\ct{e}{^\alpha_a} $, and $n^a$ is the degeneration vector field: $ \ct{n}{^a}\ctp{g}{_{ab}}=0 $. Thus, $n^a$ is the null generator of $\scri$. It follows from our choice of gauge that the second fundamental form of $ \scri $, $ \ct{K}{_{ab}}:= \ct{e}{^\mu_a}\ct{e}{^\nu_b} \cd{_\mu}\ct{n}{_\nu} $ (which is intrinsic to $\scri$, as $\lied_{\vec{n}} \bar{g}_{ab} =2K_{ab}$), vanishes. This implies that the Levi-Civita connection of $g_{\alpha\beta}$ induces a torsion-free connection on $ \scri $ by means of \cite{Geroch1977,Ashtekar81}
$$
				 \forall X,Y\in T_\scri (M), \hspace{5mm} \cds{_X} Y := \nabla_X Y 
$$
or equivalently $\ct{e}{^\alpha_a}\nabla_\alpha \ct{e}{^\beta_b}=\bar{\gamma}^c_{ab} \ct{e}{^\beta_c}$, $\bar{\gamma}^c_{ab}$ being the connection coefficients in the chosen basis. A volume form $\epsilon_{abc}$ for $\scri$ is also defined through the one of the space-time, $ \ct{\eta}{_{\alpha\beta\gamma\delta}} $, by (e.g.  \cite{Mars_Senovilla_1993})	
			\begin{equation}
			-\ct{n}{_\alpha}\ct{\epsilon}{_{abc}}\eqs \ct{\eta}{_{\alpha\mu\nu\sigma}}\ct{e}{^\mu_a}\ct{e}{^\nu_b}\ct{e}{^\sigma_c}\quad .\label{vol}
			\end{equation}	
		The choice of gauge implies $ \cds{_a}\ct{n}{^b}= 0 $ and one also has
			\begin{equation}
			\cds{_c}\bar g_{ab} =0, \hspace{1cm} \lied_{\vec n} \bar g_{ab}=0, \hspace{1cm}  \cds{_d}\epsilon_{abc} =0 \label{der=0}
			\end{equation}
so that the connection is volume preserving. This will allow us to use the Gauss law at $ \scri $ to present the balance equation \eqref{balanceint1} below.

		The topology of $ \scri $ is $ \mathbb{R}\times \mathbb{S}^2$ \cite{hawking_ellis_1973,Penrose68}, and it has two disconnected parts representing future and past null infinity. Our results apply to both of them but, for the sake of shortness and clarity, we will just consider outgoing gravitational radiation at future null infinity, the past results are simply analogous by time reversal. A cross-section $\Sc$ on $\scri$ is any closed (compact with no boundary) surface transversal to $n^a$ everywhere; we call these cross-sections `cuts', and they are topological spheres $\mathbb{S}^2$ and spacelike surfaces in $(M,g_{\alpha\beta})$. Due to (\ref{der=0}) all possible cuts are isometric, with a 1st fundamental form which is essentially the non-degenerate part of $\bar{g}_{ab}$; such inherited positive-definite metric on each $\Sc$ is denoted by $q_{AB}$ (capital Latin indices are used on the cuts, $ A,B,\dots=2,3 $), and its covariant derivative by $D_A$. 
	Given any cut $\Sc$, there is a {\em unique} lightlike vector field $\ell^\alpha$ 
		orthogonal to $\Sc$ and normalised as $ \ct{n}{_\mu} \ct{\ell}{^\mu}\eqc -1$. Thus, $\{\vec n,\vec\ell\}$ is a basis of the normal space to the cut. Let $\{E^\mu_A\}$ denote a couple of linearly independent vector fields tangent to $\Sc$. Since $ \ct{E}{^\mu_A}\ct{n}{_\mu}\eqc 0 $, one has $\ct{E}{^\mu_A}\eqc\ct{E}{^a_A}\ct{e}{^\mu_a}$, where $\{ \ct{E}{^a_A} \}$ is a basis of vector fields tangent to $\Sc$ considered within $\scri$.

		In order to put our results in context it is convenient to recall the definitions and properties of the News tensor and of the Bel-Robinson tensor. We devote the next subsections to this purpose ---for further details see \cite{Geroch1977} and \cite{Senovilla2000a}, respectively, and references therein. 
		 
	\subsection{The News tensor}			
		Concerning the news tensor field, usually denoted by $ \ct{N}{_{ab}} $, one of its possible definitions is the projection to $ \scri $ of the Schouten tensor $S_{\mu\nu}:=(R_{\mu\nu}-(R/6)g_{\mu\nu})/2$, {\em gauge corrected} \cite{Geroch1977,Ashtekar81}. It is a symmetric and gauge-invariant tensor field on $\scri$ satisfying $ \ct{n}{^a}\ct{N}{_{ab}}= 0 $. In general, however, $\lied_{\vec n} N_{ab} \neq 0$, so that $N_{ab}$ may change from one cut to another, which is a key point. Given a cut $\Sc\subset \scri$, the pullback of the news tensor to $\Sc$ will be denoted\footnote{As stated, $N_{AB}$ depends on the cut; generally, we do not use any label to make this explicit to avoid a messy notation, but we will use it when it is necessary.}  by $ \ct{N}{_{AB}} \eqc N_{ab}E^a{}_A E^b{}_B$, which is a symmetric and  traceless ($ \ct{q}{^{AB}}\ct{N}{_{AB}} = 0 $) tensor field on $\Sc$. We will use the notation
		\be
		\dot N_{AB}\defeqc E^a{}_A E^b{}_B \lied_{\vec n} N_{ab}\label{Ndot}
		\ee
		and $q^{AB}\dot N_{AB}=0$ as is easily seen.
		
		The \emph{classical radiation condition} states that  there is no gravitational radiation on a given cut if and only if the News tensor, or equivalently $ \ct{N}{_{AB}} $, vanishes on that cut:
			\begin{equation*}
				N_{AB} =0 \iff \ct{N}{_{ab}}\eqc 0 \iff \text{no gravitational radiation on } \Sc.
			\end{equation*}
Observe that, $N_{ab}$ being a tensor field on $\scri$, its vanishing at any given point is independent of any basis and thus a fully local ---pointwise--- statement. Nevertheless, as argued in \cite{Penrose1986}, one cannot aspire to {\em localize} gravitational energy at a point on $\scri$. Hence, the vanishing of $N_{ab}$ at a given point has no meaning in principle, but its vanishing on a closed surface, {\em on a cut}, does. In this sense, the news tensor field is related to the quasilocal energy-momentum properties of the gravitational field at $\scri$.
			
		\subsection{The Bel-Robinson tensor}	
		Consider now the Bel-Robinson tensor, defined by \cite{Bel1958,Bel1962,BonillaSenovilla97,Senovilla2000a}
			\begin{equation}\label{BR-tensor}
				\ct{\T}{_{\alpha\beta\gamma\delta}}\defeq \ct{C}{_{\alpha\mu\gamma}^\nu}\ct{C}{_{\delta\nu\beta}^\mu} + \ctr{^*}{C}{_{\alpha\mu\gamma}^\nu}\ctr{^*}{C}{_{\delta\nu\beta}^\mu} 
			\end{equation}	
			where $\ct{C}{_{\alpha\mu\gamma}^\nu}$ is the Weyl tensor on $(M,g_{\mu\nu})$ and $\ctr{^*}{C}{_{\alpha\mu\gamma}^\nu}:=(1/2)\eta_{\alpha\mu\rho\sigma}C^{\rho\sigma}{}_{\gamma}{}^\nu$ its dual. $\ct{\T}{_{\alpha\beta\gamma\delta}}$ is a conformally invariant, fully symmetric and traceless tensor field, and it is also divergence-free in Einstein spaces; it is assumed to describe tidal energy-momentum properties of the pure gravitational field. Given an arbitrary unit, future-pointing vector field $ \ct{v}{^\alpha} $, one can define the super-momentum relative to $v^\alpha$
			\begin{equation}\label{super-momentum}
				\ct{\Pc}{^\alpha} \defeq -  \ct{v}{^\beta}\ct{v}{^\gamma}\ct{v}{^{\delta}}\ct{\T}{^\mu_{\beta\gamma\delta}}= \cs{W}\ct{v}{^\alpha} + \ct{P}{^\alpha}	
			\end{equation}
			whose timelike component $W$ along $v^\alpha$ is called the super-energy density and its spatial part (with respecto to $v^\alpha$) $P^\alpha$ is called the super-Poynting vector field, relative to $v^\alpha$  \cite{Bel1962,Maartens1998,Senovilla2000a,Alfonso2007}:
			\begin{eqnarray}
				\cs{W}&\defeq &\ct{v}{^\alpha}\ct{v}{^\beta}\ct{v}{^\gamma}\ct{v}{^\delta}\ct{\T}{_{\alpha\beta\gamma\delta}}\quad ,\label{W}\\
				\ct{P}{^\alpha} &\defeq &-(\delta_\nu^\alpha + \ct{v}{^\alpha}\ct{v}{_\mu})\ct{v}{^\beta}\ct{v}{^\gamma}\ct{v}{^{\delta}}\ct{\T}{^\mu_{\beta\gamma\delta}}\quad \label{P}.	
			\end{eqnarray}
		The Bel-Robinson tensor satisfies a dominant property \cite{Senovilla2000a} which implies that $\ct{\Pc}{^\alpha}$ is always causal and future pointing, so that in particular the super-energy density is always non-negative $W\geq 0$. Actually, its vanishing is equivalent to the absence of the Weyl tensor at any point of the space-time:
			\begin{equation}
				\ct{\T}{_{\alpha\beta\gamma\delta}} = 0 \liff \cs{W}=0 \liff \ct{C}{^\alpha_{\beta\gamma\delta}} =0\quad .
			\end{equation}	
		Concerning the vanishing of the super-Poynting vector $P^\alpha$, Bel gave the following criterion \cite{Bel1962}:
		\begin{deff}
			There is a state of intrinsic gravitational radiation at a point $ p $ when $ \ct{P}{^\alpha}|_p \neq 0$ for all unit timelike $\ct{v}{^\alpha}  $.
		\end{deff}
		Equivalently, there is no intrinsic state of gravitational radiation at a point $p$  if there exists a unit timelike $v^\alpha$ such that $\ct{P}{^\alpha}|_p =0$. Observe that this criterion is purely local, valid at any space-time point. The idea behind the criterion is that $\ct{\Pc}{^\alpha}$ provides the space-time direction of propagation of gravity for the observer $v^\alpha$, and thus $\ct{\Pc}{^\alpha}$ proportional to $v^\alpha$ implies that there is only super-energy, but no {\em spatial} momentum. For more details on Bel's criterion see \cite{FerrandoSaez}. Even though the super-momentum cannot be identified with a momentum vector for gravity pointwise, there is a relation between super-momentum and quasilocal momentum, for instance on closed surfaces \cite{Horowitz82,Szabados2004}, that we would like to exploit. 
		As we will presently see, our proposal inherits the spirit of Bel's criterion and leads to an intrinsic characterization of the presence of gravitational radiation at infinity completely equivalent to the classical one. 
		
	\section{Radiant super-momentum}	
		Since we want to study the gravitational field at infinity, and the standard Bel-Robinson tensor \eqref{BR-tensor} vanishes at $ \scri$ --due to the known property $ \ct{C}{_{\beta\gamma\delta}^\alpha}\eqs 0 $ \cite{Geroch1977}--, we introduce the \emph{rescaled Bel-Robinson tensor}:
			\begin{equation}
				\ct{\D}{_{\alpha\beta\gamma\delta}} \defeq \Omega^{-2}\ct{\T}{_{\alpha\beta\gamma\delta}} =  \ct{d}{_{\alpha\mu\gamma}^\nu}\ct{d}{_{\delta\nu\beta}^\mu} + \ctr{^*}{d}{_{\alpha\mu\gamma}^\nu}\ctr{^*}{d}{_{\delta\nu\beta}^\mu}
			\end{equation}
		where $ \ct{d}{_{\beta\gamma\delta}^\alpha} \defeq \Omega^{-1}\ct{C}{_{\beta\gamma\delta}^\alpha}$ is the rescaled Weyl tensor.
				$\ct{\D}{_{\alpha\beta\gamma\delta}}$ is conformally invariant, fully symmetric, traceless and, using the property $\nabla_\rho d_{\mu\nu\tau}{}^\rho \stackrel{\hat M}{=} (1/\Omega)\hat\nabla_\rho \hat C_{\mu\nu\tau}{}^\rho$, its divergence vanishes at $ \scri $ if the Cotton tensor $2\hat\nabla_{[\alpha}\hat S_{\beta]\gamma}$ on $(\hat M,\hat g_{\alpha\beta})$ has appropriate decaying conditions towards infinity. $ \ct{\D}{_{\alpha\beta\gamma\delta}}  $ is regular at $ \scri $, non-vanishing in general. Its gauge behaviour is 
				$$\ct{\D}{_{\alpha\beta\gamma\delta}}\rightarrow \ct{\D}{_{\alpha\beta\gamma\delta}}/\omega^2.$$ 

		One can argue that there may be incoming radiation at $ \scri $ propagating along $ \ct{n}{^a} $ \cite{Newman1968}. Hence, if one wishes to study gravitational radiation at $\scri$ with the standard super-momentum  (\ref{super-momentum}) for an observer $ \ct{v}{^\alpha} $, this will contain information about incoming and outgoing radiation. Choose then a family of accelerated observers whose velocity vectors $ \ct{u}{^\alpha} $ approach (incoming) null cones as they are further away, and thus they will `reach' $ \scri $, in the limit to infinity, such that $ \ct{u}{^\alpha} $ becomes collinear with $ \ct{n}{^\alpha} $ there. Then, this `limit' observer  will be lightlike and tangent to $ \scri $. Qualitatively, we could say that the limit observer travels alongside the incoming radiation and thus she is insensible to it. Therefore, the unique contribution to the gravitational radiation that she could detect comes from the  outgoing components of the gravitational field. 
		This motivates us to define the fundamental object in this work, that we call the \emph{radiant super-momentum}:
			\begin{equation}\label{radiant-super-momentum}
				\ct{\Q}{^\alpha}\defeq -\ct{n}{^\mu}\ct{n}{^\nu}\ct{n}{^\rho}\ct{\D}{^\alpha_{\mu\nu\rho}}\quad .
			\end{equation}
		This definition is similar to that of the standard super-momentum \eqref{super-momentum}, however, there are three important remarks to be made. The first one is that the radiant super-momentum is defined using the rescaled Bel-Robinson tensor, which makes it regular and in general non-vanishing at $ \scri $, while \eqref{super-momentum} vanishes identically at $ \scri $. The second one is that $ \ct{n}{_\alpha} $ is lightlike at $ \scri $ and, therefore, (\ref{radiant-super-momentum}) corresponds to a `lightlike decomposition' of the rescaled Bel-Robinson tensor at $ \scri $, in contrast to the usual $ 3+1 $ splitting. Finally, and perhaps its most distinguishing attribute, the radiant super-momentum is geometrically well and uniquely defined at $\scri$, $ \ct{n}{^\alpha} $ being tangent to the null generators of $ \scri $; in contrast, \eqref{super-momentum} is observer dependent. Remarkably, it can be shown 
that, as a consequence of the peeling theorem \cite{Sachs1962,Penrose1986}, every super-momentum vector field (\ref{super-momentum}) on the physical space-time $(\hat M,\hat g_{\alpha\beta})$ converges in direction to $ \pt{\Q}{^\alpha}\defeq  -\ct{n}{^\mu}\ct{n}{^\nu}\ct{n}{^\rho}\pt{\D}{^\alpha_{\mu\nu\rho}} $ in the limit to infinity. This by itself strongly suggests that $ \ct{\Q}{^\alpha} $ contains information relative to the states of intrinsic gravitational radiation in the limit to infinity.
		
		Some further relevant properties of the radiant super-momentum are: 
			\begin{enumerate}
					\item \label{Q-future-property} $  \ct{\Q}{^\mu}$ is null $  \ct{\Q}{^\mu}\ct{\Q}{_\mu} \eqs 0  $ and future pointing at $\scri$, as follows from known properties of Bel-Robinson tensors \cite{Bergqvist_2004,Senovilla2000a}.\footnote{This is most easily seen using spinors \cite{Penrose1986}, because in spinorial form $\ct{\D}{_{\alpha\beta\gamma\delta}}\leftrightarrow d_{ABCD} \bar{d}_{A'B'C'D'}$ where $d_{ABCD}$ is the fully symmetric spinor equivalent to $d_{\mu\nu\tau}{}^\rho$.}
					\item $ \ct{\Q}{_\alpha} = \pt{\Q}{_\alpha}  $ and $ \ct{\Q}{^\alpha} = \Omega^{-2}\pt{\Q}{^\alpha}  $ on $ \ps{M} $.
					\item Under gauge transformations it transforms as 
					\be \ct{\Q}{^\alpha} \rightarrow \omega^{-7} \left(\ct{\Q}{^\alpha} -3\frac{\Omega}{\omega} D^\alpha{}_{\beta\rho\tau}n^\beta n^\rho \nabla^\tau\omega\right) +O(\Omega^2).\label{Qgauge}
					\ee
					\item \label{divergence-free-property} If the energy-momentum tensor of the physical space-time $(\hat M,\hat g_{\mu\nu})$ behaves approaching $\scri$ as $ \pt{T}{_{\alpha\beta}}|_\scri \sim \mathcal{O}(\Omega^3)$ (which includes the vacuum case  $ \pt{T}{_{\alpha\beta}}=0 $), then 
					$$ 
					\cd{_\mu}\ct{\Q}{^\mu}\eqs 0 . 
					$$
			\end{enumerate}
		From \eqref{Qgauge} one can derive the gauge change $\nabla_\mu \Q^\mu \stackrel{\scri}{\rightarrow}  \omega^{-7}\nabla_\mu\Q^\mu$, so that the statement of property \ref{divergence-free-property} is gauge invariant and leads to a balance law at $ \scri $, see section \ref{sec:balance}. The condition on the energy-momentum tensor can be replaced by a pure geometrical condition, namely that the rescaled Cotton tensor vanishes at $ \scri $, which avoids using any field equations.

		On a given cut $ \Sc $, one can split the radiant super-momentum into its null transverse (along $\ell^\alpha$) and tangent parts to $ \scri $,
			\begin{equation}
			\ct{\Q}{^\alpha}\eqc \cs{\W}\ct{\ell}{^\alpha} + \ctp{\Q}{^\alpha}=\cs{\W}\ct{\ell}{^\alpha} + \ctp{\Q}{^a}\ct{e}{^\alpha_a}\quad ,
			\end{equation}
		where
			\begin{eqnarray}\label{defRadiantSP}
			\cs{\W}&:=& -\ct{n}{_\mu}\ct{\Q}{^\mu} \quad ,\\
			\ctp{\Q}{^a}&\defeqc& \cs{\Z}\ct{n}{^a} + \ctp{\Q}{^A}\ct{E}{^a_A}\quad \text{with}
			\hspace{1cm} \cs{\Z}\defeqc -\ct{\ell}{_\mu}\ct{\Q}{^\mu} \geq 0\quad .\label{Z}
			\end{eqnarray}
		Notice that these quantities are analogous to the standard super-Poynting (\ref{P}) and super-energy density (\ref{W}) but with the advantage of being observer-independent and corresponding to a lightlike decomposition: $\Z$ and $\ctp{\Q}{^A}$ depend only on the cut, while $\W$ is fully intrinsic to $\scri$. We will refer to $ \cs{\W} |_\scri$ as the \emph{radiant super-energy density} and to $ \ctp{\Q}{^a}|_{\cal{S}} $ as the \emph{radiant super-Poynting vector} at $\cal S$. The former satisfies $ \W|_\scri\geq 0 $ and vanishes only if $\ct{\Q}{^\alpha}$ is aligned with $n^\alpha$, as follows from property \ref{Q-future-property}. 
		Actually, note that property \ref{Q-future-property} implies $\ctp{\Q}{^A}\ctp{\Q}{_A}\eqc2\cs{\W}\cs{\Z}$, so that the vanishing of either $\cs{\Z}$ or $\cs{\W}$ implies $\ctp{\Q}{^A}=0$.

		The Bianchi identity for the Weyl tensor, projected to $ \scri $ and written in terms of the News tensor 
		implies \cite{Geroch1977,Ashtekar81}
		\begin{equation}
			2\cds{_{[a}}\ct{N}{_{b]c}}=-H_{abc}, \hspace{1cm} H_{abc}:\eqs  e^\alpha{}_a e^\beta{}_b e^\gamma{}_c\, \,  d_{\alpha\beta\gamma}{}^\mu n_\mu.
			\end{equation}
		By using this equation and some properties associated to the lightlike decomposition of the rescaled Bel-Robinson tensor, a straightforward calculation 
provides the relation between the radiant super-momentum and the News tensor on a given cut $ \Sc $ (using (\ref{Ndot})): 
		\begin{align}
			\cs{\W}&\eqc 2\dot {N}^{RT}\dot {N}_{RT} \geq 0\quad ,\label{nQ}\\
			\cs{\Z}&\eqc 4 \cdc{^{[M}}\ct{N}{^{N]C}}\cdc{_{[M}} \ct{N}{_{N]C}}=2 D_CN^C{}_A D_B N^{BA}  \geq 0\quad ,\label{lQ}\\
			\ctp{\Q}{^A} &\eqc 	 8\dot N_{MC}\cdc{^{[M}} \ct{N}{^{A]C}} = -4 \dot N^{MA} D_E N^E{}_M\quad .\label{FQ}	
			\end{align}
		It is significant that the radiant super-momentum contains the information quadratic in the first derivatives of the News tensor.

	\section{The radiation condition}
		As our main result, we prove a theorem characterising the presence of radiation at $\scri$:
			\begin{thm}[Radiation condition]\label{th}
				 There is no gravitational radiation on a given cut $ \Sc \subset \scri $ if and only if the radiant super-Poynting $ \ctp{\Q}{^a} $ vanishes on that cut:\\
				 \begin{equation*}
				   \ct{N}{_{AB}}= 0 \quad\liff\quad \ctp{\Q}{^a}\eqc 0 \quad (\liff\quad \cs{\Z}= 0).
				 \end{equation*}
				
			\end{thm}
			\begin{proof}
				Consider equation \eqref{lQ}. Since the right-hand side is a square, it follows that $\cs{\Z}= 0 \iff \cdc{_{[A}} \ct{N}{_{B]C}} =  0 $. Using now equation \eqref{FQ} together with (\ref{Z}), this happens if and only if $ \ctp{\Q}{^a}= 0 $. But $\cdc{_{[A}} \ct{N}{_{B]C}} =  0$ ---which is equivalent to $D_A N^A{}_B =0$--- states that $ \ct{N}{_{AB}} $ is a symmetric and traceless Codazzi tensor on the compact 2-dimensional $\Sc$, and then it necessarily vanishes (e.g. \cite{Liuetal} and references therein). Equivalently, $N_{AB}$ is a traceless symmetric divergence-free tensor on the closed $\Sc$, which implies that $N_{AB}=0$. Hence $  \ct{N}{_{AB}}= 0  \iff \ctp{\Q}{^a}= 0 $ on $\Sc$ .
			\end{proof}
			\begin{remark}\label{thremark}
				This theorem can be equivalently stated, using property \ref{Q-future-property}, as: {\em there is no gravitational radiation on a given cut $ \Sc \subset \scri $ if and only if the radiant super-momentum is orthogonal to $\Sc$ everywhere and not co-linear with $n^\alpha$}. That is to say: $\ct{N}{_{AB}}= 0 \iff  \ct{\Q}{^\alpha}\eqc \cs{\W}\ct{\ell}{^\alpha}$. 
				Notice that, given a cut, this statement is totally unambiguous.
			\end{remark}

		Therefore, the presence of gravitational radiation is characterised equivalently with $ \ct{\Q}{^\alpha} $  or with $ \ct{N}{_{ab}} $. 
		Let $\Delta\subset\scri$ denote an open portion of $\scri$ with topology $\mathbb{S}^2\times\mathbb{R}$. 
\begin{thm}[No radiation on $ \Delta$] \label{norad2} 
There is no gravitational radiation on the open portion $ \Delta \subset \scri $ if and only if the radiant super-momentum $ \Q{^\alpha} $ vanishes on $\Delta$:
				 \begin{equation*}
				   \ct{N}{_{ab}}\stackrel{\Delta}{=}  0 \quad\liff\quad \Q{^\alpha}\stackrel{\Delta}{=}  0 .
				 \end{equation*}
\end{thm}	
\begin{proof}
According to remark \ref{thremark} of theorem \ref{th}, absence of radiation on $\Delta$ requires that $\ct{\Q}{^\alpha}\eqc \cs{\W}\ct{\ell}{^\alpha}$ on {\em every possible cut} $\Sc$ included in $\Delta$. But this is only possible if $\ct{\Q}{^\alpha}\stackrel{\Delta}{=} 0$. Another route to derive this result is to note that $N_{AB}=0$ on every cut within $\Delta$, and thus $N_{ab}\stackrel{\Delta}{=} 0$. In particular $\lied_{\vec n} N_{ab}\stackrel{\Delta}{=} 0$ so that $\dot N_{AB}$ vanishes too at any cut within $\Delta$. 
\end{proof}

\begin{remark}
				The condition $ \Q{^\alpha}\stackrel{\Delta}{=} 0 $ implies \cite{Bergqvist1998,Senovilla2011} that $ \ct{n}{^\alpha} $ is a multiple principal null direction of $ \ct{d}{_{\beta\gamma\delta}^\alpha} $ on $\Delta$, which is in accordance with the discussion in \cite{Krtou2004}.
			\end{remark}
			\begin{remark}
A practical good property of the novel characterisation provided by theorems \ref{th} and \ref{norad2} is that the condition $\Q{^\alpha}= 0$ is computation friendly, avoiding the calculation of the news tensor.
\end{remark}

\noindent	
 \subsection{The case of sandwich waves}\label{sandwich1}	
	Let now $\Delta\subset\scri$ be an open portion of $\scri$ as before, connected and limited by two cuts $\Sc_1$ and $\Sc_2$, with $\Sc_2$ entirely to the future of $\Sc_1$.	
	We assume that $ \ct{N}{_{ab}} $ is at least continuous on $\scri$.
		A sandwich wave is defined by the existence of radiation inside the open connected portion $ \Delta$, but not outside it (locally). This case is characterised by non-vanishing functions $\Z$, equivalently by a non-vanishing $\Q{^\alpha}$, on $\Delta$. 
		
				 The generic situation for sandwich waves has $\W\neq 0$ everywhere on $\Delta$. However, a special situation arises if  $\W$ vanishes somewhere on $\Delta$. On such regions the radiant super-momentum aligns with the generators of $ \scri $, and $n^\alpha$ is a simple principal null direction of $ \ct{d}{_{\beta\gamma\delta}^\alpha} $ there. 
				Notice that this  alignment occurs where $\lied_{\vec{n}}N_{ab}=0$ . 
				Observe that, by continuity, $N_{ab}$ vanishes at the two boundary cuts $\Sc_1$ and $\Sc_2$, but one must have $\lied_{\vec{n}}N_{ab}\neq 0$ (ergo $\dot N_{AB}\neq 0$) there.

	\section{A balance law at null infinity}\label{sec:balance}
		As a second result, we present a balance law at $ \scri $ at the super-energy level. We use the same notation as before for the connected $ \Delta \subset\scri $ and its bounding cuts $ \Sc_{1} $ and $ \Sc_{2} $. On $\Sc_{1,2}$ we have the corresponding null vector fields $\ell^\alpha_{1}$ and $\ell^\alpha_2$ and news ${}_1 N_{AB}$ and ${}_2N_{AB}$. We denote by $L^\alpha$ {\em any} null vector field defined on $\Delta$ such that 
		\be
		n_\alpha L^\alpha \stackrel{\Delta}{=} -1, \hspace{1cm} L^\alpha\stackrel{\Sc_{1}}{=} \ell^\alpha_{1}, \hspace{1cm} L^\alpha\stackrel{\Sc_{2}}{=} \ell^\alpha_{2}.\label{Ldef}
		\ee
		 Observe that there is a large freedom to choose $L^\alpha$.  The computation of the projected derivative of $ \ct{\Q}{^\alpha} $ at $\Delta\subset  \scri $, together with property \ref{divergence-free-property}, leads to\footnote{One may be tempted to rewrite the left-hand side in (\ref{balancediff}) as $\nabla_\mu (\W L^\mu)$, but this would require extending $L^\mu$ outside $\scri$, for instance geodesically: $L^\mu\nabla_\mu L^\nu=0$.}
			\begin{equation}\label{balancediff}
				L{^\mu}\cd{_\mu}\cs{\W} + \cs{\W}\ct{\psi}{^a_a} \stackrel{\Delta}{=}  -\cds{_a}\ctp{Q}{^a}
			\end{equation}
			where $\psi^a_b$ is a tensor field on $\Delta$ defined by $e^\mu{}_a\nabla_\mu L^\nu \stackrel{\Delta}{:=} \psi^b_a e^\nu{}_b$.
Equation \eqref{balancediff} has the form of a continuity law and, moreover, we can integrate and bring it into a flux form:
			\begin{equation}\label{balanceint1}
			 	\int_{\Delta}\left(L{^\mu}\cd{_\mu}\cs{\W} + \cs{\W}\ct{\psi}{^a_a}\right) \boldsymbol{\epsilon}= \int_{\Sc_2}	\Z \boldsymbol{\mathring{\epsilon}} -  \int_{\Sc_1} \Z \boldsymbol{\mathring{\epsilon}}
			\end{equation}
			where $\boldsymbol{\mathring{\epsilon}}$ is the canonical volume 2-form associated to $q_{AB}$ on the cuts. Interestingly, the left-hand side represents the flux of radiant super-energy \emph{escaping} from the space-time in any outgoing null direction $ L{^\alpha} $ and this is controlled by the positive quantity $ \int_{\Sc}	\cs{\Z} $ evaluated at the boundaries of $\Delta$. 
			 Note that the left-hand side is independent of the choice of $L^\mu$. To see this, take any other possible $\tilde L^\mu$ subject to (\ref{Ldef}), so that necessarily $s^\mu :=\tilde L^\mu -L^\mu $ is such that $n_\mu s^\mu =0$ ergo $s^\mu =s^ae^\mu{}_a$ with $s^a\stackrel{\Sc_{1,2}}{=}0$. Hence $\tilde\psi^b_a =\psi^b_a +\overline\nabla_as^b$ and then
			$$
			\int_{\Delta}\left(\tilde L{^\mu}\cd{_\mu}\cs{\W} + \cs{\W}\ct{\tilde\psi}{^a_a}\right) \boldsymbol{\epsilon} -\int_{\Delta}\left(L{^\mu}\cd{_\mu}\cs{\W} + \cs{\W}\ct{\psi}{^a_a}\right) \boldsymbol{\epsilon}=\int_\Delta \overline\nabla_a (\W s^a) =0.
			$$
			On the other hand, neither of the integrals in the balance formula \eqref{balanceint1} are gauge invariant (as the integrands change by a factor $\omega^{-4}$), nevertheless \eqref{balanceint1} holds true in any gauge.
		
		For the situation in theorem \ref{norad2}, the balance law is trivially satisfied. For sandwich waves as defined in section \ref{sandwich1}, 
		the integrals on the right-hand side vanish and, therefore,
			\begin{equation}\label{balanceforsandwich}
				\int_{\Delta}\left(L{^\mu}\cd{_\mu}\cs{\W} + \cs{\W}\ct{\psi}{^a_a}\right)\boldsymbol{\mathring{\epsilon}} = 0\quad \mbox{for sandwich waves} .
			\end{equation}
		This is actually always the case when ${}_1N_{AB}={}_2N_{AB}$ (as all cuts are isometric, this equality makes sense).

	\subsection{Relation to the Bondi-Trautman energy loss
	}\label{dimensions-BR-tensor}
	By inserting equation \eqref{lQ} into the right-hand side of (\ref{balanceint1})
and manipulating the integrand a little, it is possible to write the following alternative version:
			\begin{equation}\label{balanceint2}
				\int_{\Delta}\left(L{^\mu}\cd{_\mu}\cs{\W} + \cs{\W}\ct{\psi}{^a_a}\right)\boldsymbol{\epsilon}= \int_{\Sc} \ct{N}{_{AB}}\left(2K\ct{N}{^{AB}}-\cdc{_C}\cdc{^C}\ct{N}{^{AB}}\right)\boldsymbol{\mathring{\epsilon}}\, \Bigg|_{\Sc_1}^{\Sc_2}
			\end{equation}
		where $ K $ is the Gaussian curvature of the cuts. 
		
		Let $F\in C^\infty (\scri)$ be any smooth function on $\scri$ such that $\varphi:=\lied_{\vec n} F\neq 0$ everywhere. Any such $F$ defines a foliation of $\scri$ by smooth cuts, each cut given by $F=C$ with constant $C$. Let $\cs{E}[][BT] (C)$ denote the Bondi-Trautman energy \cite{Bondi1962,Geroch1977,CJM,Szabados2004} of these cross sections relative to a chosen infinitesimal translation \cite{Geroch1977} $\alpha n^a$, where $ \cs{\alpha} $ is a function which necessarily satisfies $ \lied_{\vec{n}}\alpha \eqs 0 $. In our conventions, the infinitesimal Bondi-Trautman energy loss on a cut is given by 
			\begin{equation}\label{beloss}
				\frac{d\cs{E}[][BT]}{d C} = -\frac{1}{8\pi}\int_{\Sc}\frac{\cs{\alpha}}{ \varphi}\ct{N}{_{AB}}\ct{N}{^{AB}}\boldsymbol{\mathring{\epsilon}} \quad .
			\end{equation}
		  Without restricting the gauge freedom to any particular choice, it is always possible to adapt the foliation to the generators of $ \scri $ by setting $ \lied_{\vec{n}}\varphi = 0 $. Comparing \eqref{beloss} with the first term on the right-hand side of formula \eqref{balanceint2}, we see that the Bondi-Trautman energy-loss term appears in the balance law  for cuts satisfying $K\varphi/\alpha=$ constant ---in particular, this holds by choosing an appropriate gauge, which is always possible after selecting an adapted foliation. Under such a choice (\ref{balanceint2}) can be written as
		\begin{equation}\label{balanceint3}
				\int_{\Delta}\left(L{^\mu}\cd{_\mu}\cs{\W} + \cs{\W}\ct{\psi}{^a_a}\right)\boldsymbol{\epsilon}=\left[-16\pi K\frac{\varphi}{\alpha} \frac{d\cs{E}[][BT]}{d C}  -\int_{\Sc} \ct{N}{_{AB}}\cdc{_C}\cdc{^C}\ct{N}{^{AB}}\boldsymbol{\mathring{\epsilon}}\right]\, \Bigg|_{\Sc_1}^{\Sc_2} \, .
			\end{equation}
		This nice formula allows us to perform a quick check of the physical units\footnote{We use the notation $ [P] $ to denote the {\em physical} units of any object $ P $; our choice is that the conformal factors $\Omega$ and $\omega$ are dimensionless.}: $ [\cs{E}[][BT]] = ML^2T^{-2} $ so that $ [d\cs{E}[][BT]/d C] =ML^2T^{-2} [C]^{-1}$. As $ [\varphi]=[C]L^{-2} $ and taking into account $ \cs{[\alpha]} = L $ the right-hand side of (\ref{balanceint3}) has dimensions of $[K\varphi/\alpha]ML^2T^{-2} [C]^{-1}=ML^{-3}T^{-2}$. Concerning the left-hand side, using that $[L^\mu]=L$ and that $[\W]=\left[\ct{\T}{_{\alpha\beta\gamma\delta}} \right] L^{-4}$, we need to know the units of the volume integral on $\scri$ but, according to \eqref{vol}, these are $[\boldsymbol{\epsilon}]=L^4$. Hence, $ \left[\ct{\D}{_{\alpha\beta\gamma\delta}} \right]= MT^{-2}L^{-3} $ and the physical units of the Bel-Robinson tensor are
			$$
			 \left[\ct{\T}{_{\alpha\beta\gamma\delta}} \right] = MT^{-2}L^{-3}\quad .
		 $$
		Even though the units of $ \ct{\D}{_{\alpha\beta\gamma\delta}} $ depend on the dimensions of the conformal factor, the result for the Bel-Robinson tensor is independent of this choice. This is in agreement with \cite{Horowitz82} and \cite{Teyssandier1999mg}, see also \cite{Szabados2004,Senovilla2000a}, because \emph{the Bel-Robinson tensor has physical dimensions of energy density per unit surface}, so that to recover physical units in the equations where it appears one should write
			\begin{equation*}
				\frac{c^4}{G}\ct{\T}{_{\alpha\beta\gamma\delta}}\quad .
			\end{equation*}		
			
	\section{Discussion}
		We have introduced the radiant super-momentum, which is an observer-independent, geometrically well-defined quantity at $\scri$, containing information quadratic in the first derivatives of the News tensor. We have proven that an alternative characterisation of gravitational radiation at null infinity exists at the super-energy level, and that it is completely equivalent to the usual News tensor criterion. Moreover, we have obtained a new balance law at infinity which describes the outgoing flux of super-energy density due to the presence of gravitational waves arriving to a region of $ \scri $. Remarkably, this flux depends on the News tensor on the cuts that bound this region. This fact suggests defining some general `radiation states' on the cuts, the flux of super-energy between two sections of $  \scri $ representing the fail of the system to recover its initial state. Equation \eqref{balanceforsandwich} may be seen as a condition for the system to recover its initial state.
		
		We will extend these ideas in forthcoming work and, particularly, we will address the $ \Lambda>0 $ case in \cite{Fernandez-Alvarez_Senovilla}.
			\subsection*{Acknowledgments}
				This work is supported under Grants No. FIS2017-85076-P (Spanish MINECO/AEI/FEDER, UE) and No. IT956-16 (Basque Government).
		\bibliography{novel_characterisation_of_radiation_afs_letter}{}
	\end{document}